\documentclass[a4paper,11pt]{article}
\usepackage{authblk}
\usepackage{todonotes}
\usepackage{multirow}
\usepackage{etoolbox}
\usepackage{amsmath, amssymb, amsthm,complexity}
\usepackage{graphicx}
\usepackage{xcolor}
\usepackage{tcolorbox}
\tcbuselibrary{skins} 
\usepackage{algorithm}
\usepackage{thmtools}
\usepackage{thm-restate}
\usepackage[noend]{algpseudocode}
\newtheorem{mytheorem}{Theorem}
\newtheorem{mylemma1}{Lemma}
\newtheorem{mycorollary}{Corollary}
\newtheorem{mylemma}{Lemma}
\newtheorem{myproposition}{Proposition}
\newtheorem{myobservation}{Observation}
\newtheorem{mydefinition}{Definition}

\usepackage{amsmath}
\usepackage{amssymb}
\usepackage{amsthm}
\usepackage{graphicx}
\usepackage{complexity} 
\usepackage{todonotes} 
\usepackage{enumitem}
\setlist{nosep}

\usepackage[hang]{footmisc}

\newcommand{\yes}{\textsc{Yes}}
\newcommand{\no}{\textsc{No}}

\newtheorem{observation}{Observation}
\newtheorem{reduction rule}{Reduction Rule}[section]
\newtheorem{marking-scheme}{Marking Scheme}[section]

\usepackage{mathtools}
\usepackage{bold-extra}

\usepackage{fullpage}

\algnewcommand{\LeftComment}[1]{\Statex \(\triangleright\) #1}

\algrenewcommand\algorithmicrequire{\textbf{Input:}}
\algrenewcommand\algorithmicensure{\textbf{Output:}}

\newcommand{\parprob}[4]{
  \vspace{2mm}
\noindent\fbox{
  \begin{minipage}{0.96\textwidth}
  \begin{tabular*}{\textwidth}{@{\extracolsep{\fill}}lr} \textsc{#1}  & {\bf{Parameter:}} #3
\\ \end{tabular*}
  {\bf{Input:}} #2  \\
  {\bf{Question:}} #4
  \end{minipage}
  }
  \vspace{2mm}
}

\title{Balanced Substructures in Bicolored Graphs\footnote{A preliminary version of this work appears in the proceedings of the 48th International Conference on Current Trends in Theory and Practice of Computer Science (SOFSEM 2023).\hfill} \thanks{The second author is supported by SERB CRG grant number CRG/2022/007751.\hfill}}

\author[1]{P.~S.~Ardra\footnote{111914001@smail.iitpkd.ac.in~}}

\author[1]{R.~Krithika\footnote{krithika@iitpkd.ac.in~}}

\author[2,3]{Saket~Saurabh\footnote{saket@imsc.res.in~}}

\author[4]{Roohani Sharma\footnote{rsharma@mpi-inf.mpg.de}}

\affil[1]{Indian Institute of Technology Palakkad, Palakkad, India.}
\affil[2]{The Institute of Mathematical Sciences, Chennai, India.}
\affil[3]{University of Bergen, Bergen, Norway.}
\affil[4]{Max Planck Institute for Informatics,Saarland Informatics Campus, Saarbrucken, Germany.}

\theoremstyle{plain}
\usepackage{hyperref}
\date{}

\begin{document}
\maketitle
\begin{abstract}
An edge-colored graph is said to be {\em balanced} if it has an equal number of edges of each color. 
Given a graph $G$ whose edges are colored using two colors and a positive integer $k$, the objective in the \textsc{Edge Balanced Connected Subgraph} problem is to determine if $G$ has a balanced connected subgraph containing at least $k$ edges. We first show that this problem is \NP-complete\ and remains so even if the solution is required to be a tree or a path. Then, we focus on the parameterized complexity of \textsc{Edge Balanced Connected Subgraph} and its variants (where the balanced subgraph is required to be a path/tree) with respect to $k$ as the parameter. Towards this, we show that if a graph has a balanced connected subgraph/tree/path of size at least $k$, then it has one of size at least $k$ and at most $f(k)$ where $f$ is a linear function. We use this result combined with dynamic programming algorithms based on {\em color coding} and {\em representative sets} to show that \textsc{Edge Balanced Connected Subgraph} and its variants are \FPT. Further, using polynomial-time reductions to the \textsc{Multilinear Monomial Detection} problem, we give faster randomized \FPT\ algorithms for the problems. In order to describe these reductions, we define a combinatorial object called {\em relaxed-subgraph}. We define this object in such a way that balanced connected subgraphs, trees and paths are relaxed-subgraphs with certain properties. This object is defined in the spirit of branching walks known for the \textsc{Steiner Tree} problem and may be of independent interest.
\end{abstract}


\section{Introduction}

\label{sec:intro}
Ramsey Theory is a branch of Combinatorics that deals with patterns in large arbitrary structures. 
In the context of edge-colored graphs where each edge is colored with one color from a finite set of colors, a fundamental problem in the area is concerned with the existence of \textit{monochromatic} subgraphs of a specific type. Here, monochromatic means that all edges of the subgraph have the same color. 
For simplicity, we discuss only undirected graphs where each edge is colored either red or blue. 
Such a coloring is called a \textit{red-blue coloring} and a graph associated with a red-blue coloring is referred to as a \textit{red-blue graph}.
In this work, we study questions related to the existence of and finding \emph{balanced} subgraphs instead of monochromatic subgraphs, where by a balanced subgraph we mean one which has an equal number of edges of each color. These problems come under a subarea of Ramsey Theory known as Zero-Sum Ramsey Theory. Here, given a graph whose vertices/edges are assigned weights from a set of integers, one looks for conditions that guarantee the existence of a certain subgraph having total weight zero.
For example, one may ask when is a graph whose all the edges are given weight -1 or 1 guaranteed to have a spanning tree with total weight of its edges 0. This is equivalent to asking when a red-blue graph is guaranteed to have a balanced spanning tree. Necessary and sufficient conditions have been established for complete graphs, triangle-free graphs and maximal planar graphs \cite{CaroHLZ22}.
In the same spirit, one may ask a more general question like when is a red-blue graph $G$ guaranteed to have a balanced connected subgraph of size (number of edges) $k$. An easy necessary condition is that there are at least $k/2$ red edges and at least $k/2$ blue edges in $G$. 
This condition is also sufficient (as we show in the proof of Theorem \ref{thm:bcs-split}) if $G$ is a complete graph (or more generally a split graph).
However, we do not think that such a simple characterization will exist for all graphs. 
This brings us to the following natural algorithmic question concerning balanced connected subgraphs.

\parprob{\textsc{Edge Balanced Connected Subgraph}}{A red-blue graph $G$ and a positive integer $k$}{$k$}{Does $G$ have a balanced connected subgraph of size at least $k$?}

When the subgraph is required to be a tree or a path, the corresponding variants of \textsc{Edge Balanced Connected Subgraph} are called \textsc{Edge Balanced Tree} and \textsc{Edge Balanced Path}, respectively. 
We show that these problems are \NP-complete. 
\begin{itemize}
    \item (Theorems \ref{thm:bcs-npc}, \ref{thm:bt-npc}, \ref{thm:bp-npc}) \textsc{Edge Balanced Connected Subgraph}, \textsc{Edge Balanced Tree} and \textsc{Edge Balanced Path} are \NP-complete.
\end{itemize}
In fact, \textsc{Edge Balanced Connected Subgraph} and \textsc{Edge Balanced Tree} remain \NP-complete\ on bipartite graphs, planar graphs and chordal graphs.
However, \textsc{Edge Balanced Connected Subgraph} is polynomial-time solvable on split graphs (Theorem \ref{thm:bcs-split}).
Yet, \textsc{Edge Balanced Path} is \NP-complete\ even on split graphs.

Note that if a graph has a balanced connected subgraph/tree/path of size at least $k$, then it is not guaranteed that it has one of size equal to $k$. 
This brings us to the following combinatorial question: if a graph has a balanced connected subgraph/tree/path of size at least $k$, then can we show that it has a balanced connected subgraph/tree/path of size equal to $f(k)$ for some function $f$?
We answer these questions in the affirmative and show the existence of such a function which is linear in $k$. 

\begin{itemize}
    \item (Theorems \ref{thm:small-bal-path}, \ref{thm:small-bal-tree}, \ref{thm:small-bal-sub}) If a graph has a balanced connected subgraph/tree of size at least $k$, then it has one of size at least $k$ and at most $3k+3$. Further, if a graph has a balanced path of size at least $k$, then it has a balanced path of size at least $k$ and at most $2k$.
\end{itemize}
Therefore, in order to find a balanced connected subgraph/tree/path of size at least $k$, it suffices to focus on the problem of finding a balanced connected subgraph/tree/path of size exactly $k$.
This leads us to the following problem. 

\parprob{\textsc{Exact Edge Balanced Connected Subgraph}}{A red-blue graph $G$ and a positive integer $k$}{$k$}{Does $G$ have a balanced connected subgraph of size $k$?}

As before, when the connected subgraph is required to be a tree or a path, the corresponding variants of \textsc{Exact Edge Balanced Connected Subgraph} are called \textsc{Exact Edge Balanced Tree} and \textsc{Exact Edge Balanced Path}, respectively. These problems are also \NP-complete\ and so we study them from the perspective of parameterized complexity. 

In this framework, the key notion is that of a {\em parameterized problem} which is a decision problem where each input instance is associated with a non-negative integer $\ell$ called the {\em parameter}.  A parameterized problem is said to be {\em fixed-parameter tractable} if it can be solved in $f(\ell) n^{\mathcal{O}(1)}$ time for some computable function $f$ where $n$ is the input size. Algorithms with such running times are called \FPT\ algorithms or parameterized algorithms and the complexity class \FPT\ denotes the set of all parameterized problems that are fixed-parameter tractable. For convenience, the running time $f(\ell) n^{\mathcal{O}(1)}$ of a parameterized algorithm where $f$ is a super-polynomial function is written as $\mathcal{O}^*(f(\ell))$. A {\em kernelization} or {\em kernel} for a parameterized problem is a polynomial-time algorithm that transforms an arbitrary instance of the problem to an equivalent instance of the same problem whose size is bounded by some computable function $g$ of the parameter of the original instance. 
A kernel is a {\em polynomial kernel} if $g$ is a polynomial function and we say that the problem admits a polynomial kernel. 
A {\em polynomial parameter transformation} from problem $\Pi_1$ to $\Pi_2$ is a polynomial-time algorithm that transforms an arbitrary instance of $\Pi_1$ to an equivalent instance of $\Pi_2$ such that the parameter of the resulting instance is bounded by some polynomial function of the parameter of the original instance. A parameterized problem is \FPT\ if and only if it is decidable and has a kernel. 
Therefore, the set of parameterized problems that admit polynomial kernels is contained in the class \FPT\ and it is believed that this subset relation is strict. Polynomial parameter transformations are useful in ruling out polynomial kernels. For more information on parameterized complexity, we refer to the book by Cygan et al. \cite{CyganFKLMPPS15}. 

Focussing on \textsc{Exact Edge Balanced Connected Subgraph/ Tree/Path}  with respect to the solution size $k$ as the parameter, we give randomized \FPT\ algorithms for solving the three problems using reductions to the \textsc{Multilinear Monomial Detection} problem (defined in Section \ref{sec:fpt-alg}). 
\begin{itemize}
    \item (Theorems \ref{thm:bcs-fpt-rand}, \ref{thm:bt-fpt-rand}, \ref{thm:bp-fpt-rand}) \textsc{Exact Edge Balanced Connected Subgraph/ Tree/Path} can be solved by a randomized algorithm in $\mathcal{O}^*(2^k)$ time.
\end{itemize}

Many problems reduce to \textsc{Multilinear Monomial Detection} \cite{KoutisW16} and the current fastest algorithm solving it is a randomized algorithm that runs in $\mathcal{O}^*(2^{k})$ time \cite{KoutisW16,Koutis08,Williams09}.
The reductions that we give to \textsc{Multilinear Monomial Detection} use a combinatorial object called {\em relaxed-subgraph}. This object is defined in the spirit of \textit{branching walks} known for the \textsc{Steiner Tree} problem  \cite{Nederlof13}.
We define this object in such a way that balanced connected subgraphs, trees and paths are relaxed-subgraphs with certain properties. Then, using the {\em color-coding} technique \cite{CyganFKLMPPS15,AlonYZ95} and {\em representative sets} \cite{CyganFKLMPPS15,FominLPS16,ShachnaiZ16}, we give deterministic dynamic programming algorithms for the problems. 
\begin{itemize}
    \item (Theorems \ref{thm:bcs-fpt}, \ref{thm:bt-fpt}, \ref{thm:bp-fpt-rep}) \textsc{Exact Edge Balanced Connected Subgraph/ Tree} can be solved in $\mathcal{O}^*((4e)^k)$ time and \textsc{Exact Edge Balanced Path} can be solved in $\mathcal{O}^*(2.619^k)$ time.
\end{itemize}
The method of representative sets is a generic approach for designing efficient dynamic programming based parameterized algorithms that may be viewed as a deterministic-analogue to the color-coding technique. 
Representative sets have been used to obtain algorithms for several parameterized problems \cite{FominLPS16} and our algorithm adds to this list. \\

\smallskip
\noindent 
{\bf Road Map. }
The \NP-completeness of the problems are given in Section \ref{sec:np-c-results}. 
In Section \ref{sec:com-lemma}, the combinatorial results related to the existence of small balanced connected subgraphs, trees and paths are proven.
Section \ref{sec:fpt-alg} discusses the deterministic and randomized algorithms for the problems. 
Section \ref{sec:concl} concludes the work by listing some future directions.\\

\smallskip
\noindent 
{\bf Related Work. }A variant of \textsc{Exact Edge Balanced Connected Subgraph} has recently been studied \cite{Bhore0JMPR19,BhoreJPR19,DartiesGKP19,KobayashiKMSY19,MartinodPDGK21}.
In order to state these results using our terminology, we define the notion of {\em vertex-balanced subgraphs} of {\em vertex-colored graphs}.
A coloring of the vertices of a graph using red and blue colors is called a \textit{red-blue vertex coloring}.

A subgraph of a vertex-colored graph is said to be \textit{vertex-balanced} if it has an equal number of vertices of each color. In the \textsc{Exact Vertex Balanced Connected Subgraph} problem, the interest is in finding a vertex-balanced connected subgraph on $k$ vertices in the given graph associated with a red-blue vertex coloring. This problem is \NP-complete\ and remains so on restricted graph classes like bipartite graphs, planar graphs, chordal graphs, unit disk graphs, outer-string graphs, complete grid graphs, and unit square graphs \cite{Bhore0JMPR19,BhoreJPR19}. 
However, polynomial-time algorithms are known for trees, interval graphs, split graphs, circular-arc graphs and permutation graphs \cite{Bhore0JMPR19,BhoreJPR19}. 
Further, the problem is \NP-complete\ even when the subgraph required is a path \cite{Bhore0JMPR19}. 
\FPT\ algorithms, exact exponential-time algorithms and approximation results for the problem are known from \cite{BhoreJPR19}, \cite{KobayashiKMSY19} and \cite{MartinodPDGK21}.
Observe that finding vertex-balanced connected subgraphs in vertex-colored graphs reduces to finding vertex-balanced trees while the analogous solution in edge-colored graphs may have more complex structures. 

\section{Preliminaries}
Let $[k]$ denote the set $\{1,2,...,k\}$ for $k \in \mathbb{N}_+$. For standard graph-theoretic terminology not stated here, we refer the reader to the book by Diestel \cite{Diestel12}. 
For an undirected graph $G$, $V(G)$ denotes its set of vertices and $E(G)$ denotes its set of edges. 
The {\em size} of a graph is the number of its edges and the {\em order} of a graph is the number of its vertices. An edge between vertices $u$ and $v$ is denoted as $\{u,v\}$ and $u$ and $v$ are called the {\em endpoints} of the edge $\{u,v\}$. Two vertices $u, v$ in $V(G)$ are \emph{adjacent} if $\{u,v\} \in E(G)$. The \emph{neighborhood} of a vertex $v$, denoted by $N_G(v)$, is the set of vertices adjacent to $v$. Similarly, two edges $e, e'$ in $E(G)$ are \emph{adjacent} if they have exactly one common endpoint and the \emph{neighborhood} of an edge $e$, denoted by $N_G(e)$, is the set of edges adjacent to $e$. The {\em degree} of a vertex $v$ is the size of $N_G(v)$. 
A vertex is called an {\em isolated vertex} if its degree is 0. An edge with an endpoint that has degree 1 is called a  {\em pendant edge}. The notion of neighborhood is extended to a set $S \subseteq V(G)$ of vertices by defining $N_G(S)$ as $(\bigcup_{v \in S} N(v)) \setminus S$. We omit the subscript in the notation for neighborhood if the graph under consideration is clear. For a set $F \subseteq E(G)$ of edges, $V(F)$ denotes the set of endpoints of edges in $F$. For a set $S \subseteq V(G) \cup E(G)$, $G - S$ denotes the graph obtained by deleting $S$ from $G$. For a set $S \subseteq V(G)$ of vertices, the subgraph of $G$ induced on $S$ is denoted by $G[S]$. For two sets $S, T \subseteq V(G)$ of vertices, $E(S, T)$ denotes the set of edges with one endpoint in $S$ and the other endpoint in $T$. 

A {\em walk} in $G$ is a sequence $(v_1,\dots,v_k)$ of vertices such that for each $i \in [k-1]$, $\{v_i,v_{i+1}\} \in E(G)$.
A {\em path} $P$ in $G$ is a walk $(v_1,\dots,v_k)$ consisting of distinct vertices. 
For a path $P=(v_1,\dots,v_k)$, the set $\{v_1,\dots,v_k\}$ is denoted by $V(P)$ and the set $\{\{v_i,v_{i+1}\} \mid i \in [k-1]\}$ is denoted by $E(P)$. 
We say that $P$ {\em starts at} $v_1$ and {\em ends at} $v_k$. Vertices $v_1$ and $v_k$ are called the {\em endpoints} of $P$ and edges $\{v_1,v_2\}$ and $\{v_{k-1},v_{k}\}$ are called {\em terminal edges} of $P$.
The {\em length} of a path or walk is the number of edges in it. 
A graph is {\em connected} if there is a path between every pair of its vertices. 
Given a graph $G$, its {\em line graph} $L(G)$ is defined as $V(L(G))=\{e \mid e \in E(G)\}$ and $E(L(G))=\{\{e,e'\} \mid e \text{ and } e' \text{ are adjacent in } G\}$. 
It is well-known that a graph $G$ without isolated vertices is connected if and only if $L(G)$ is connected. A \textit{tree} is an undirected connected acyclic graph. A \textit{clique} is a set of pairwise adjacent vertices and a \textit{complete graph} is a graph whose vertex set is a clique. 
An \textit{independent set} is a set of pairwise non-adjacent vertices. A \textit{split graph} is a graph whose vertex set can be partitioned into a clique and an independent set. 
 
\section{NP-hardness Results}
\label{sec:np-c-results} 

We first show the \NP-hardness of \textsc{Edge Balanced Connected Subgraph} using a polynomial-time reduction from the well-known \NP-hard\ \textsc{Steiner Tree} problem \cite[ND12]{GareyJ79}. In this problem, given a connected graph $G$, a subset $T\subseteq V(G)$ (called {\em terminals}) and a positive integer $k$, the task is to determine if $G$ has a subtree $H$ (called a {\em Steiner tree}) with $T\subseteq V(H)$ and $|E(T)| \leq k$. 
The idea behind the reduction is to color all edges of $G$ of the \textsc{Steiner Tree} instance blue and add exactly $k$ red edges incident to the terminals such that each terminal has at least one red edge incident on it.
Any balanced connected subgraph of size (at least) $k$ of the resulting graph is required to include all the red edges and hence includes all the terminals which in turn corresponds to a Steiner tree of $G$.

\begin{restatable}{mytheorem}{bcsnpc}
\label{thm:bcs-npc}
\textsc{Edge Balanced Connected Subgraph} is \NP-complete. 
\end{restatable}

\begin{proof}
It is easy to verify that \textsc{Edge Balanced Connected Subgraph} is in \NP. Consider an instance $(G,T,k)$ of \textsc{Steiner Tree}. Let  $T=\{t_1,t_2,\ldots,t_{\ell}\}$ be the terminal vertices. Without loss of generality assume $\ell\leq k$ and $\vert E(G) \vert \geq k$. Now we construct an instance $(H,2k)$ of \textsc{Edge Balanced Connected Subgraph} that is equivalent to $(G,T,k)$.
The graph $H$ is obtained from $G$ as follows: for each $t\in T$ we add a new vertex $t'$ adjacent to $t$. 
Then we add $k-\ell$ new vertices adjacent to $t_1$. 
Formally $V(H)=V(G)\cup \{{t'} \mid t \in T\} \cup \{{t'_{\ell+i}} \mid i\in[k-\ell] \} $ and $E(H)= E(G) \cup \{\{t, t'\} \mid t \in T\} \cup  \{\{t_1, {t'_{\ell+i}}\} \mid i\in[k-\ell] \}$. 
Next we define a red-blue coloring of $E(H)$ as follows: edges in $E(H) \cap E(G)$ are colored blue and edges in $E(H) \setminus E(G)$ are colored red.
Let $T' = V(H) \backslash V(G)$ and $E'= E(H) \setminus E(G)$. 

Suppose $(H,2k)$ is a \yes-instance.  Let $H'$ be a balanced connected subgraph of $H$ with at least $2k$ edges. Then, as $H$ has exactly $k$ red edges (which is the set $E'$), $\vert E(H') \vert=2k$. 
Further, as $V(E')= T \cup T'$ it follows that $T \cup T' \subseteq V(H')$. 
Now consider the subgraph $H''$ of $G$ obtained from $H'$ by deleting vertices in $T'$. Clearly, $H''$ has exactly $k$ edges. 
As every vertex in $T'$ has degree 1 and $N_H(T')\subseteq V(G)$, $H''$ is connected. 
Further, $T\subseteq V(H'')$. 
Now any spanning tree of $H''$ is a Steiner tree with at most $k$ edges. Thus $(G,T,k)$ is a \yes-instance.

Conversely,  suppose $(G,T,k)$ is a \yes-instance. Then let $T^*$ be a subtree of $G$ containing $T$ with $\vert E(T^{*}) \vert \leq k$. Let $T^{**}$ be a subgraph of $G$, containing $T^{*}$ such that $\vert E(T^{**}) \vert =k$. Note that $T^{**}$ is also a subgraph of $H$ and each edge in this subgraph is blue. 
Define another subgraph $H'$ of $H$ containing $T^{**}$ as follows: 
$V(H')=V(T^{**}) \cup T'$ and $E(H') = E(T^{**}) \cup E'$. 
Clearly, $H'$ is a balanced connected subgraph of $H$ of size $2k$. 
Thus, $(H,2k)$ is a \yes-instance.
\end{proof}

As the variant of the \textsc{Steiner Tree} problem where a tree on exactly $k$ edges is required is also \NP-complete\, we have the following result.

\begin{restatable}{mytheorem}{btnpc}
\label{thm:bt-npc}
\textsc{Edge Balanced Tree} is \NP-complete. 
\end{restatable}

As the reduction described in Theorem \ref{thm:bcs-npc} is a polynomial parameter transformation, the infeasibility of the existence of polynomial kernels for \textsc{Steiner Tree} parameterized by the solution size (i.e., the size of the tree) \cite{CyganFKLMPPS15,DomLS14} extend to  \textsc{Edge Balanced Connected Subgraph} and \textsc{Edge Balanced Tree} as well.
Further, since \textsc{Steiner Tree} has no subexponential \FPT\ algorithm assuming the Exponential Time Hypothesis, it follows that \textsc{Edge Balanced Connected Subgraph} and \textsc{Edge Balanced Tree} also do not admit subexponential \FPT\ algorithms.
Moreover, the reduction in Theorem \ref{thm:bcs-npc} preserves planarity, bipartiteness and chordality. 
This property along with the \NP-completeness\ of \textsc{Steiner Tree} (and its variant) on bipartite graphs \cite{GareyJ79}, planar graphs \cite{GareyJ77} and chordal graphs \cite{WhiteFP85} imply that \textsc{Edge Balanced Connected Subgraph} and \textsc{Edge Balanced Tree} are \NP-complete\ on planar graphs, chordal graphs and bipartite graphs as well. 

\subsection{Complexity in Split Graphs}
Next, we consider \textsc{Edge Balanced Connected Subgraph} on split graphs.
Let $(G,k)$ be an instance. 
An easy necessary condition for $G$ to have a balanced connected  subgraph of size (at least) $k$ is that there are at least $k/2$ red edges and at least $k/2$ blue edges in $G$. 
We show that this condition is also sufficient if $G$ is a split graph leading to the following result.

\begin{restatable}{mytheorem}{bcssplit}
\label{thm:bcs-split}
\textsc{Edge Balanced Connected Subgraph} is polynomial-time solvable on split graphs.
\end{restatable}
\begin{proof}
Let $(G,k)$ be an instance where $G$ is a split graph whose vertex set is partitioned into clique $C$ and independent set $I$.
Let $E_R$ be the set of red edges and $E_B$ be the set of blue edges in $G$. If $\vert E_R \vert < \frac{k}{2}$ or $\vert E_B \vert < \frac{k}{2}$, then $(G,k)$ is a \no-instance. Otherwise, $\vert E_R \vert \geq \frac{k}{2}$ and $\vert E_B \vert \geq \frac{k}{2}$. We show that $(G,k)$ is a \yes-instance.

Let $v$ be a vertex in $C$. Define sets $X =\{x \mid x\in C, \{v,x\} \in E_R\}$ and $Y=\{y \mid y\in C, \{v,y\} \in E_B\}$. We will now construct a balanced connected subgraph $H$ of $G$ in each of the following four cases.
\begin{itemize}
    \item Case 1: $\vert X \vert \geq \frac{k}{2}$ and $\vert Y \vert \geq \frac{k}{2}$\\
    Consider any set $S_X$ of $\frac{k}{2}$ vertices from $X$ and any set $S_Y$ of $\frac{k}{2}$ vertices from $Y$. Then, the subgraph $H$ of $G$ with $V(H)=S_X \cup S_Y \cup \{v\}$ and $E(H)=E(v,S_X) \cup E(v,S_Y)$ is the required balanced connected subgraph.
    
    \item Case 2: $\vert X \vert < \frac{k}{2}$ and $\vert Y \vert < \frac{k}{2}$\\
    Initialize $E'=\emptyset$. Let $\vert X \vert=k_1$ and $\vert Y \vert=k_2$. First, we add all edges in $E(v,X) \cup E(v,Y)$ to $E'$. Observe that $C \subseteq V(E')$. Next, we add $\frac{k}{2}-k_1$ edges from $E_R \setminus E(v,X)$ and $\frac{k}{2}-k_2$ edges from $E_B \setminus E(v,Y)$ to $E'$. Then, the subgraph $H$ of $G$ with $E(H)=E'$ and $V(H)=V(E')$ is the required balanced connected subgraph.
   
    \item Case 3: $\vert X \vert \geq \frac{k}{2}$ and $\vert Y \vert < \frac{k}{2}$\\
First, we add all edges in  $E(v,Y)$ to $E'$. Next, we add edges in $E(Y,I) \cap E_B$ to $E'$ one by one until $|E' \cap E_B| =\frac{k}{2}$ or $E(Y,I) \cap E_B \subseteq  E'$. Similarly, we add edges in $(E(v,I) \cup E(Y,Y)) \cap E_B$ to $E'$ one by one until $|E' \cap E_B| =\frac{k}{2}$ or $(E(v,I) \cup E(Y,Y)) \cap E_B \subseteq  E'$. Then, we add edges in $E(X,Y) \cap E_B$ to $E'$ one by one until $|E' \cap E_B| =\frac{k}{2}$ or $E(X,Y) \cap E_B \subseteq  E'$. When we add such an edge $\{x,y\}$ with $x \in X, y \in Y$, we also add the edge $\{x,v\}$ (which is guaranteed to be in $E_R$) to $E'$. At the end of this procedure, observe that $|E' \cap E_R| \leq \frac{k}{2}$ and $|E' \cap E_B| = \frac{k}{2}$. Finally, we add some set of $\frac{k}{2}-|E' \cap E_R|$ edges from $E(v,X)$ to $E'$. Then, the subgraph $H$ of $G$ with $E(H)=E'$ and $V(H)=V(E')$ is the required balanced connected subgraph.
  
    \item Case 4: $\vert X \vert < \frac{k}{2}$ and $\vert Y \vert  \geq \frac{k}{2}$\\
    The construction of $H$ is similar to the one in Case 3.  
  \end{itemize}  
  As the above mentioned cases are exhaustive and in each of them, we have found a procedure to find a balanced connected subgraph $H$ (which is guaranteed to exist), it follows that \textsc{Edge Balanced Connected Subgraph} is polynomial-time solvable on split graphs.
\end{proof}

Now, we move on to \textsc{Edge Balanced Path} and show that it is \NP-hard\ on split graphs by giving a polynomial-time reduction from \textsc{Longest Path}. 
In the \textsc{Longest Path} problem, given a graph $G$ and a positive integer $k$, the task is to find a path $P$ in $G$ of length $k$. 
It is known that \textsc{Longest Path} is \NP-hard \cite[ND29]{GareyJ79} and remains so on split graphs even when the starting vertex $u_0$ of the path is given as part of the input \cite[GT39]{GareyJ79}. 
The reduction may be viewed as attaching a red path of length $k$ (consisting of new internal vertices) starting from $u_0$ to the graph $G$ (whose edges are colored blue) of the \textsc{Longest Path} instance.
Then, by adding certain additional edges (colored blue), we get a red-blue split graph.

\begin{restatable}{mytheorem}{bpsplit}
\label{thm:bp-npc}
\textsc{Edge Balanced Path} is \NP-complete\ on split graphs.
\end{restatable}

\begin{proof}
It is easy to verify that this problem is in \NP.
Let $(G,u_0,k)$ be an instance of \textsc{Longest Path} where $G$ is a split graph whose vertex set is partitioned into clique $C$ and independent set $I$. 
Suppose $u_0\in C$ (the construction when $u_0\in I$ is similar). We construct an instance $(H,2k)$ of \textsc{Edge Balanced Path} that is equivalent to $(G,u_0,k)$. The graph $H$ is obtained from $G$ as follows. 
Initialize $V(H)=V(G)$ and $E(H)=E(G)$. Add new vertices $\{u_1,u_2,\ldots,u_k\}$ that form the path $P^*=(u_1,u_2,\ldots,u_k)$ to $G$ such that $u_1$ is adjacent to $u_0$. 
Let $S=\{u_{2i} \mid 1\leq i \leq \lfloor \frac{k}{2} \rfloor  \}$. Add edges between every pair of vertices in $S$. Also, make every vertex in $C$ adjacent to every vertex in $S$. 
Observe that $V(H)$ is partitioned into clique $C \cup S$ and independent set $I \cup (V(P^*) \setminus S)$.
Color $E(H)$ such that all edges in $(E(H)\cap E(P^*)) \cup \{u_0,u_1\}$ are red and the remaining edges are blue. 
This completes the construction of $H$.
Let $E_R$ denote the set of red edges and $E_B$ denote the set of blue edges in $H$.
Observe that $|E_R|=k$. 

Suppose $(G,u_0,k)$ is a \yes-instance. Then there exists a path $P=(u_0,x_1,x_2,$ $\ldots,x_{k})$ in $G$ of length $k$ starting with vertex $u_0$. Now, the path $P'=(u_k,\ldots,u_1,$ $u_0,x_1,x_2,\ldots,x_{k})$ is a balanced path in $H$ on $2k$ vertices. Thus $(H,2k)$ is a \yes-instance. 
Conversely, suppose $(H,2k)$ is a \yes-instance. Let $P'$ be a path in $H$ that has $k$ red edges and $k$ blue edges. 
Observe that $P^*$ is a subpath of $P'$ as $|E_R|=k$. Then, the graph obtained by deleting $V(P^*)$ from $P'$ is a path $P$ in $G$ starting at $u_0$. Further, $P$ has $k$ edges. 
Thus $(G,u_0,k)$ is a \yes-instance. 
\end{proof}

As \textsc{Longest Path} parameterized by the solution size (i.e., the size of the path) in general graphs does not admit a polynomial kernel \cite{CyganFKLMPPS15,BodlaenderDFH09} and the reduction described (which is adaptable for general graphs) is a polynomial parameter transformation, it follows that \textsc{Edge Balanced Path} does not admit polynomial kernels.
Further, it is known that, assuming the Exponential Time Hypothesis, \textsc{Longest Path} has no subexponential \FPT\ algorithm in general graphs. 
Hence, \textsc{Edge Balanced Path} also does not admit subexponential \FPT\ algorithms.

\section{Small Balanced Paths, Trees and Connected Subgraphs}
\label{sec:com-lemma}
In this section, we prove the combinatorial result that if a graph has a balanced connected subgraph/tree/path of size at least $k$, then it has one of size at least $k$ and at most $f(k)$ where $f$ is a linear function. 
We begin with balanced paths.

\begin{restatable}{mytheorem}{smallbp}
\label{thm:small-bal-path}
Let $G$ be a red-blue graph and $k \geq 2$ be a positive integer. Then, if $G$ has a balanced path of length at least $2k$, then $G$ has a smaller balanced path of length at least $k$.
\end{restatable}
\begin{proof}
Let $E_B$ be the set of blue edges and $E_R$ be the set of red edges in $G$.
Consider a balanced path $P$ in $G$ with at least $2k$ edges. 
If the terminal edges $e$ and $e'$ are of different colors, then delete $e$ and $e'$ to get a smaller path of length at least $k$. 
Otherwise, let $P=(v_1,v_2,\ldots,v_\ell)$ where $e_i$ denotes the edge $\{v_i,v_{i+1}\}$ for each $i \in [\ell-1]$. Without loss of generality, let $e_1, e_{\ell-1} \in E_R$. 
Define the function $h:E(P)\rightarrow \mathbb{N}$ as follows.
\[
h(e_i) = 
\begin{cases} 
1, & \text{if } i =1\\
h(e_{i-1})+1, & \text{if } i >1 \text{ and } e_i \in E_R\\
        h(e_{i-1})-1, & \text{if } i >1 \text{ and } e_i \in E_B
\end{cases}
\]
Let $E'$ be the edges $e$ in $P$ with $h(e)=0$. 
Clearly, $|E'| \geq 1$ as $P$ is balanced and so $h(e_{\ell-1})=0$.
We claim that $|E'|>1$.
Suppose $|E'|=1$. 
Then, as $e_{\ell-1} \in E_R$, $h(e_{\ell-1})=h(e_{\ell-2})+1$ implying that $h(e_{\ell-1})=-1$.
As $h(e_1)=1$ and $h(e_{\ell-1})=-1$, it follows that there is an edge $e_i$ with $i \neq \ell$ and $h(e_i)=0$. This contradicts that $|E'|=1$.
Hence, $|E'|>1$.
Let $e_i$ denote an edge with $i<\ell-1$ and $h(e_i)=0$. 
Then, the subpaths $P_1$ and $P_2$ with $E(P_1)=\{e_1, \ldots, e_i\}$ and $E(P_2)=\{e_{i+1}, \ldots, e_{\ell-1}\}$ are two balanced paths strictly smaller than $P$.
Further, as $|E(P)| \geq 2k$, either $|E(P_1)| \geq k$ or $|E(P_2)| \geq k$.
\end{proof}

Now, we move to the analogous result for balanced trees. 

\begin{restatable}{mytheorem}{smallbt}
\label{thm:small-bal-tree}
Let $G$ be a red-blue graph and $k \geq 2$ be a positive integer. Then, if $G$ has a balanced tree with at least $3k+2$ edges, then $G$ has a smaller balanced tree with at least $k$ edges.
\end{restatable}
\begin{proof}
Let $E_B$ be the set of blue edges and $E_R$ be the set of red edges in $G$. 
Consider a balanced tree $T$ in $G$ with at least $3k+2$ edges.
If $T$ is a path, then by Theorem \ref{thm:small-bal-path}, we obtain the desired smaller tree which is a path.
If $T$ has pendant edges $e$ and $e'$ of different colors, then delete $e$ and $e'$ to get a smaller tree on at least $k$ edges. 
Otherwise, without loss of generality, let all pendant edges of $T$ be in $E_R$.
Let $n$ denote $|V(T)|$.
Root $T$ at an arbitrary vertex of degree at least 3.
For a vertex $v \in V(T)$, let $T_v$ denote the subtree of $T$ rooted at $v$.
Let $u$ be a vertex with maximum distance from the root such that $|V(T_u)| > \frac{n}{3}$. 
Let $u_1, \ldots, u_\ell$ be the children of $u$.
Observe that for each $i \in [\ell]$, $|V(T_{u_i})| \leq \frac{n}{3}$.
Let $i$ be the least integer in $[\ell]$ such that $\frac{n}{3} \leq |\underset{1 \leq j \leq i} \bigcup V(T_{u_j})| \leq \frac{2n}{3}$.
Let $S$ denote $\underset{1 \leq j \leq i} \bigcup V(T_{u_j})$ and $R$ denote $V(T) \setminus S$.
As $\frac{n}{3} \leq |S| \leq \frac{2n}{3}$, we have $\frac{n}{3} \leq |R| \leq \frac{2n}{3}$.
Further, since $n \geq 3k+3$, we have $\frac{n}{3} \geq k+1$.
Hence, $k+1 \leq |S|, |R| \leq n-1$. 
Consider the following cases.\\
Case 1: $T[S \cup \{u\}]$ or $T[R]$ is balanced. Then, the lemma holds trivially.\\
Case 2: $T[S \cup \{u\}]$ or $T[R]$ has more edges from $E_R$ than from $E_B$. Suppose $T[R]$ has more edges from $E_R$ than from $E_B$. 
Initialize $T^*$ to be $T[R]$. 
As $T[R]$ has at least $k+1$ vertices (and therefore at least $k$ edges), it has at least $k/2$ edges from $E_R$. 
Add the edges of $T[S \cup \{u\}]$ to $T^*$ in the breadth-first order until $T^*$ becomes balanced.
We claim that $T \neq T^*$.
Suppose $T=T^*$. 
Then, let $e$ be the edge that was added last to $T^*$. 
Clearly, $e$ is a pendant edge and is in $E_R$. 
The edge $e$ was added as $T^*-e$ is not balanced. 
Further, the number of edges in $T^*-e$ from $E_B$ is less than the number of edges from $E_R$. 
Otherwise, the addition process would have terminated earlier.
However, if $|E(T^*-e) \cap E_B|<|E(T^*-e) \cap E_R|$, $e \in E_R$ and $T=T^*$, then it follows that $T$ is not balanced leading to a contradiction.\\
Case 3: $T[S \cup \{u\}]$ and $T[R]$ have lesser edges from $E_R$ than from $E_B$. As $E(T[S \cup \{u\}])$ and $E(T[R])$ partition $E(T)$, this case implies that $T$ has more edges from $E_B$ than from $E_R$ contradicting that $T$ is balanced.
\end{proof}

Finally, we prove the result for balanced connected subgraphs. For this, we use line graphs, vertex-balanced subgraphs and vertex-balanced trees of vertex-colored graphs. 

\begin{restatable}{mytheorem}{smallbcs}
\label{thm:small-bal-sub}
Let $G$ be a red-blue graph and $k \geq 2$ be a positive integer. Then, if $G$ has a balanced connected subgraph with at least $3k+3$ edges, then $G$ has a smaller balanced connected subgraph with at least $k$ edges.
\end{restatable}
\begin{proof}
Define a red-blue coloring on $V(L(G))$ as follows: for each vertex $x$ in $L(G)$ corresponding to a red (blue) edge $\{u,v\}$ in $G$, color $x$ using red (blue). 
Suppose $G$ has a balanced connected subgraph $G'$ with $\ell \geq k$ edges. 
Then, $L(G')$ (which is a subgraph of $L(G)$) is connected and has exactly $\ell$ vertices with half of them colored red and the rest colored blue. 
That is, $L(G)$ has a vertex-balanced connected subgraph of order $\ell$.
Suppose $L(G)$ has a vertex-balanced connected subgraph $H$ on $\ell$ vertices. 
Let $G'$ denote the subgraph of $G$ defined as $E'=E(G')=\{\{u,v\}\in E(G) \mid e \in V(H), e=\{u,v\}\}$ and $V(G')=V(E')$. 
Then, $G'$ is a balanced connected subgraph of $G$ with $\ell$ edges. 

Now, it remains to show that if $L(G)$ has a vertex-balanced connected subgraph (equivalently, a balanced tree $T$) with at least $3k+3$ vertices, then $L(G)$ has a smaller vertex-balanced connected subgraph (equivalently, a balanced tree $T^*$) with at least $k$ vertices. The proof of this claim is similar to the proof of Theorem \ref{thm:small-bal-tree}.
\end{proof}

Observe that due to Theorems \ref{thm:small-bal-path}, \ref{thm:small-bal-tree} and \ref{thm:small-bal-sub}, it suffices to give \FPT\ algorithms for \textsc{Exact Edge Balanced Connected Subgraph/Tree/Path} in order to obtain \FPT\ algorithms for \textsc{Edge Balanced Connected Subgraph/Tree/Path}.

\section{FPT Algorithms}
\label{sec:fpt-alg}
We now describe parameterized algorithms for \textsc{Exact Edge Balanced Connected Subgraph/Tree/Path}. 

\subsection{Randomized Algorithms}
In this section, we show that \textsc{Exact Edge Balanced Connected Subgraph}, \textsc{Exact Edge Balanced Path} and \textsc{Exact Edge Balanced Tree} admit randomized algorithms that runs in $\mathcal{O}^*(2^{k})$ time. We do so by reducing the problems to \textsc{Multilinear Monomial Detection}. 
In order to define this problem, we state some terminology related to polynomials from \cite{KoutisW16}. 
Let $X$ denote a set of variables. 
A {\em monomial} of degree $d$ is a product of $d$ variables from $X$, with multiplication assumed to be commutative. 
A monomial is called {\em multilinear} if no variable appears twice or more in the product. 
A {\em polynomial} $P(X)$ over $\mathbb{Z}_+$ is a linear combination of monomials with coefficients from $\mathbb{Z}_+$. 
A polynomial contains a certain monomial if the monomial appears with a non-zero coefficient in the linear combination that constitutes the polynomial. Polynomials can be represented as {\em arithmetic circuits}  which in turn can be represented as {\em directed acyclic graphs}.  
In the \textsc{Multilinear Monomial Detection} problem, given an arithmetic circuit (represented as a directed acyclic graph) representing a polynomial $P(X)$ over $\mathbb{Z}_+$ and a positive integer $k$, the task is to decide whether $P(X)$ contains a multilinear monomial of degree at most $k$.

\begin{restatable}{myproposition}{mld}
{\em \cite{KoutisW16,Koutis08,Williams09}}
\label{prop:mld}
Let $P(X)$ be a polynomial over $\mathbb{Z}_+$ represented by a circuit. 
The \textsc{Multilinear Monomial Detection} problem for $P(X)$ can be decided in randomized $\mathcal{O}^*(2^k)$ time and polynomial space.
\end{restatable}

Our reductions from \textsc{Exact Edge Balanced Connected Subgraph/ Tree/Path} to \textsc{Multilinear Monomial Detection} crucially use the notions of a \textit{color-preserving homomorphism} (also known as an edge-colored homomorphism in the literature \cite{BrewsterDHQ05}) and {\em relaxed-subgraphs}. 

\begin{restatable}{mydefinition}{colorhom}
Given graphs $G$ and $H$ with red-blue edge colorings $\texttt{col}_G:E(G) \rightarrow \{red,blue\}$ and $\texttt{col}_H:E(H) \rightarrow \{red,blue\}$, a color-preserving homomorphism from $H$ to $G$ is a function $h: V(H) \rightarrow V(G)$ satisfying the following properties.
\begin{itemize}
\item For each pair $u,v \in V(H)$, if $\{u,v\} \in E(H)$, then $\{h(u),h(v)\} \in E(G)$.
\item For each edge $\{u,v\}$ in $H$, $\texttt{col}_H(\{u,v\})=\texttt{col}_G(\{h(u),h(v)\})$.
\end{itemize}
\end{restatable}

\begin{restatable}{mydefinition}{relaxsub}
Given a red-blue graph $G$, a relaxed-subgraph is a pair $S=(H,h)$ where $H$ is a red-blue graph and $h$ is a color-preserving homomorphism from $H$ to $G$.
\end{restatable}

The vertex set of a relaxed-subgraph $S=(H,h)$ is $V(S)=\{h(a) \in V(G) \mid a \in V(H)\}$ and the edge set of $S$ is $\{\{h(a),h(b)\} \in E(G) \mid \{a,b\} \in E(H)\}$.
We treat the vertex and edge sets of a relaxed-subgraph as multi-sets.
The size of $S$ is the number of edges in $H$ (equivalently, the size of $E(S)$).
$S$ is said to be connected if $H$ is connected and $S$ is said to be balanced if $H$ has an equal number of red edges and blue edges.
$S$ is said to be a {\em relaxed-path} if $H$ is a path and a {\em relaxed-tree} if $H$ is a tree.

Next, we have the following observation that states that relaxed-subgraphs with certain specific properties correspond to balanced connected subgraphs, trees and paths.

\begin{restatable}{myobservation}{relaxsubprop}
\label{obs:relaxed-structures}
The following hold for a red-blue graph $G$.
\begin{itemize}
    \item $G$ has a balanced connected subgraph of size $k$ if and only if there is a balanced connected relaxed-subgraph $S$ of size $k$ such that $E(S)$ consists of distinct elements.
    \item $G$ has a balanced path of size $k$ if and only if there is a balanced relaxed-path $(P,h)$ of size $k$ where $h$ is injective.
    \item $G$ has a balanced tree of size $k$ if and only if there is a balanced relaxed-tree $(T,h)$ of size $k$ where $h$ is injective.
\end{itemize}
\end{restatable}
\begin{proof}
The forward direction of the three claims are easy to verify. 
If $G$ has a balanced connected subgraph/tree/path $Q$ of size $k$, then $S=(Q,h)$ where $h: V(Q) \rightarrow V(G)$ is the identity map is the required balanced connected relaxed-subgraph/tree/path of size $k$.
Consider the converse of the first claim.
Suppose there is a balanced connected relaxed-subgraph $S=(H,h)$ of size $k$ such that $E(S)$ consists of distinct elements.
Consider the subgraph $G'$ of $G$ with $V(G')=V(S)$ and $E(G')=E(S)$.
Then, $G'$ is connected as $H$ is connected and $|E(G')|=|E(S)|=k$ as $E(S)$ consists of distinct elements. Further, since $H$ is balanced, $G'$ is balanced.
Hence, $G$ has a balanced connected subgraph of size $k$.
The proofs of other two claims are similar.
\end{proof}

Now, we are ready to describe the randomized algorithms for \textsc{Exact Edge Balanced Connected Subgraph/Tree/Path} based on Observation \ref{obs:relaxed-structures}. 
First, we consider \textsc{Exact Edge Balanced Connected Subgraph}.

\begin{restatable}{mytheorem}{bcsmld}
\label{thm:bcs-fpt-rand}
\textsc{Exact Edge Balanced Connected Subgraph} admits a randomized $\mathcal{O}^*(2^{k})$-time algorithm.
\end{restatable}

\begin{proof}
Consider an instance $(G,k)$. 
Let $E_R$ denote the set of red edges and $E_B$ denote the set of blue edges in $G$.
In order to obtain an instance of \textsc{Multilinear Monomial Detection} that is equivalent to $(G,k)$, we will define a polynomial $P$ over the variable set $\{x_e \mid e \in E(G)\}$ satisfying the following properties.
\begin{itemize}
    \item For each balanced connected relaxed-subgraph $S=(H,h)$ of size $k$ there exists a monomial in $P$ that corresponds to $S$.  
    We say that a monomial $M$ corresponds to $S$, if $M=\underset{e \in E(S)}\prod x_e$.
    \item Each multilinear monomial in $P$ corresponds to some balanced connected relaxed-subgraph $S$ of size $k$ where $E(S)$ has distinct elements. 
\end{itemize}

If $P$ is such a polynomial, then from Observation \ref{obs:relaxed-structures}, $G$ has a balanced connected subgraph of size $k$ if and only if $P$ has a multilinear monomial of degree $k$. This way, after the construction of $P$, we reduce the problem to {\sc Multilinear Monomial Detection} and use Proposition \ref{prop:mld}.
In order to construct $P$, we first construct polynomials $P_j(e,r,b)$ for each $e \in E(G)$, $j \in [k]$ and $0 \leq r,b \leq \frac{k}{2}$ with $r+b \geq 1$.
Monomials of $P_j(e,r,b)$ will correspond to connected relaxed-subgraphs $S=(H,h)$ of size $j$ such that $H$ has $r$ red edges, $b$ blue edges and $e \in E(S)$.
The construction of $P_j(e,r,b)$ is as follows. For an edge $e \in E(G)$, 

\noindent $P_1(e,1,0) = 
\begin{cases} 
x_e, & \text{if } e \in E_R\\
0 & \text{otherwise.}
\end{cases}$
and 
$
P_1(e,0,1) = 
\begin{cases} 
x_e, & \text{if } e \in E_B\\
0 & \text{otherwise.}
\end{cases}
$. 

\noindent Also, $P_j(e,r,b)=0$ if $j \neq r+b$. Now, if $e \in E_R$ and $r+b >1$, then we have

\begin{align*}
P_j(e,r,b) &= \underset{\substack{e' \in N_G(e), \ell <j\\r'+r''=r, b'+b''=b}} \sum P_{\ell}(e',r',b') P_{j-\ell}(e,r'',b'') + \underset{e' \in N_G(e)} \sum x_e P_{j-1}(e',r-1,b)
\end{align*}

\noindent and if $e \in E_B$ and $r+b >1$, then we have
\begin{align*}
P_j(e,r,b) &= \underset{\substack{e' \in N_G(e), \ell<j\\r'+r''=r, b'+b''=b}} \sum P_{\ell}(e',r',b') P_{j-\ell}(e,r'',b'') + \underset{e' \in N_G(e)} \sum x_e P_{j-1}(e',r,b-1).
\end{align*}

We now show that every multilinear monomial of $P_j(e,r,b)$ corresponds to a connected relaxed-subgraph $S=(H,h)$ of size $j$ such that $H$ has $r$ red edges and $b$ blue edges with $E(S)$ consisting of distinct elements where $e \in E(S)$. 
We prove this claim by induction on $j$. The base case is easy to verify. 
Consider the induction step. 
Suppose $e=\{u,v\} \in E_R$ (the other case is symmetric). 
Let $M$ be a multilinear monomial of $P_j(e,r,b)$ where $j>1$.
\begin{itemize}
\item Case 1: $M=x_e M'$ where $M'$ is a multilinear monomial of $P_{j-1}(e',r-1,b)$ such that $e' \in N_G(e)$.
Let $e'=\{v,w\}$.
By induction, $M'$ corresponds to a connected relaxed-subgraph $S'=(H',h')$ of size $j-1$ such that $e' \in E(S')$. 
Further, $H'$ has $r-1$ red edges and $b$ blue edges.
Also, $v \in V(h'(H'))$.
Let $z=h'^{-1}(v)$.
Observe that $z$ is well-defined due to the multilinearity of $M$. 
Note that $E(S')$ consists of distinct elements and $e \notin E(S')$. 
Let $H$ denote the graph obtained from $H'$  by adding a new vertex $z'$ adjacent to $z$ with the edge $\{z,z'\}$ colored red.
Let $h: V(H) \rightarrow V(G)$ denote the homomorphism obtained from $h'$ by extending its domain to include $z'$ and setting $h(z')=u$. 
Then, $S=(H,h)$ is a connected relaxed-subgraph that $M$ corresponds to.
\item Case 2: $M=M_1 M_2$ where $M_1$ is a multilinear monomial of $P_{j_1}(e',r',b')$ and $M_2$ is a multilinear monomial of $P_{j_2}(e,r'',b'')$ such that $e' \in N_G(e)$, $j_1,j_2 <j$, $r'' \leq r$ and $b'' \leq b$.
Let $e'=\{v,w\}$.
By induction, $M_1$ corresponds to a connected relaxed-subgraph $S_1=(H_1,h_1)$ of size $j_1$ such that $e' \in E(S_1)$. 
Similarly, $M_2$ corresponds to a connected relaxed-subgraph $S_2=(H_2,h_2)$ of size $j_2$ such that $e \in E(S_2)$. 
Further, $H_1$ has $r'$ red edges, $b'$ blue edges and $H_2$ has $r''$ red edges and $b''$ blue edges.
Also, $v \in V(h_1(H_1)) \cap V(h_2(H_2))$ and $E(S_1) \cap E(S_2)=\emptyset$.
Without loss of generality, assume that $V(H_1) \cap V(H_2)=\emptyset$ as this can be achieved by a renaming procedure.
Let $z_1=h_1^{-1}(v)$ and $z_2=h_2^{-1}(v)$.
Observe that $z_1$ and $z_2$ are well-defined due to the multilinearity of $M_1$ and $M_2$.
Now, rename $z_1$ in $S_1$ and $z_2$ in $S_2$ as $z$.
Let $H$ denote the graph with vertex set $V(H_1) \cup V(H_2)$ and edge set $E(H_1) \cup E(H_2)$.
Observe that $H$ is a connected graph. 
Let $h: V(H) \rightarrow V(G)$ denote the homomorphism obtained from $h_1$ and $h_2$ by extending the domain to the union of the domains of $h_1$ and $h_2$. Then, $S=(H,h)$ is a connected relaxed-subgraph that $M$ corresponds to.
\end{itemize}

We next show that if there is a connected relaxed-subgraph $S=(H,h)$ of size $j$ with $r$ red edges, $b$ blue edges and such that $e=\{u,v\} \in E(S)$, then there is a monomial of $P_j(e,r,b)$ that corresponds to it. We show this again using induction on $j$. Suppose $e\in E_R$ (the other case is symmetric). The base case is trivial. 
Consider the induction step ($j \geq  2$). Let $a=h^{-1}(u)$, $b=h^{-1}(v)$ and $z$ denote the edge $\{a,b\}$ of $H$.
\begin{itemize}
    \item Case 1: $H-z$ is connected.
Then, $S'=(H-z,h)$ is a connected relaxed-subgraph of size $j-1$ with $r-1$ red edges and $b$ blue edges and contains an edge $e' \in N_G(e)$.
By induction, there is a monomial $M'$ corresponding to $S'$ in $P_{j-1}(e',r-1,b)$. 
Then, the monomial $M=x_e M'$ which is in $P_j(e,r,b)$ corresponds to $S$.
\item Case 2: $H-z$ is disconnected.
Then $H$ has two components $H_a$ (containing $a$) and $H_b$ (containing $b$).
Without loss of generality let $H_a$ have at least one edge.
Let $H'_b$ denote the subgraph of $H$ obtained from $H_b$ by adding the vertex $a$ and edge $\{a,b\}$. 
Let $j_1$ and $j_2$ denote the number of edges in $H_a$ and $H'_b$, respectively.
Let $r'$ and $r''$ be the number of red edges in $H_a$ and $H'_b$, respectively.
Similarly, let $b'$ and $b''$ be the number of blue edges in $H_a$ and $H'_b$, respectively.
Then $j=j_1 +j_2, r=r'+r'', b=b'+b''$.
Let $h_a$ and $h_b$ denote the color-preserving homomorphism obtained from $h$ by restricting the domain to $V(H_a)$ and $V(H'_b)$, respectively. 
Now, $S_1=(H_a,h_a)$ and $S_2=(H'_b,h_b)$ are connected relaxed-subgraphs with the following properties. \\
(i) The size of $S_1$ is $j_1$ and $S_1$ has $r'$ red edges, $b'$ blue edges and there is an edge $e'\in E(S_1)$ incident on $h_a(a)$. \\
(ii) The size of $S_2$ is $j_2$ and $S_2$  has $r''$ red edges, $b''$ blue edges and $e \in E(S_2)$.

By induction, there is a monomial $M_1$ in $P_{j_1}(e',r',b')$ that corresponds to $S_1$ and there is a monomial $M_2$ in $P_{j_2}(e,r'',b'')$ that corresponds to $S_2$. 
Then, the monomial $M_1M_2$ which is in $P_j(e,r,b)$ corresponds to $S$.
\end{itemize}
Finally, let $P= \underset{e \in E(G)} \sum P_k(e,\frac{k}{2},\frac{k}{2})$.
Every monomial in $P$ has degree $k$. Then from the arguments above, $P$ is the desired polynomial.
To compute $P$ we need to compute $P_j(e,r,b)$ for each $j,r,b \in [k]$ and each edge $e \in E(G)$.
As these polynomials can be represented as a 
polynomial-sized arithmetic circuit, the reduction runs in polynomial time.
\end{proof}

Next, we move on to \textsc{Exact Edge Balanced Tree}. 

\begin{restatable}{mytheorem}{btmld}
\label{thm:bt-fpt-rand}
\textsc{Exact Edge Balanced Tree} admits a randomized $\mathcal{O}^*(2^{k})$-time algorithm.
\end{restatable}

\begin{proof}
Consider an instance $(G,k)$. 
Let $E_R$ denote the set of red edges and $E_B$ denote the set of blue edges in $G$.
In order to obtain an instance of \textsc{Multilinear Monomial Detection} that is equivalent to $(G,k)$, we will define a polynomial $P$ over the variable set $\{y_v : v \in V(G)\}$ satisfying the following properties.
\begin{itemize}
    \item For each balanced relaxed-tree $X=(T,h)$ of size $k$ there exists a monomial in $P$ that corresponds to $X$.  
    We say that a monomial $M$ corresponds to $X$, if $M=\underset{v \in V(X)}\prod y_v$.
    \item For each multilinear monomial $M$ in $P$, there is a balanced relaxed-tree $X=(Q,h)$ of size $k$ where $h$ is injective that $M$ corresponds to.
\end{itemize}

If $P$ is such a polynomial, then from Observation \ref{obs:relaxed-structures}, $G$ has a balanced tree of size $k$ if and only if $P$ has a multilinear monomial of degree $k+1$. This way, after the construction of $P$, we reduce the problem to {\sc Multilinear Monomial Detection} and use Proposition \ref{prop:mld}. 
In order to construct $P$, we first construct polynomials $P_j(e,r,b)$ for each $e \in E(G)$, $j \in [k]$ and $0 \leq r,b \leq \frac{k}{2}$ with $r+b \geq 1$.
Monomials of $P_j(e,r,b)$ will correspond to relaxed-trees $X=(T,h)$ of size $j$ such that $T$ has $r$ red edges, $b$ blue edges and $e \in E(X)$.
The construction of $P_j(e,r,b)$ is as follows. For an edge $e \in E(G)$,  we have $P_1(e,1,0) = 
\begin{cases} 
y_u y_v, & \text{if } e=\{u,v\} \in E_R\\
0 & \text{otherwise.}
\end{cases}$
and similarly 
$
P_1(e,0,1) = 
\begin{cases} 
y_u y_v, & \text{if } e=\{u,v\} \in E_B\\
0 & \text{otherwise.}
\end{cases}
$. 

\noindent Also, $P_j(e,r,b)=0$ if $j \neq r+b$. Now, for $j >1$, if $e=\{u,v\} \in E_R$, then
\begin{multline*}
P_j(e,r,b) = \underset{\substack{e',e'' \in N_G(e)\\ u \in V(e'), v \in V(e'')\\ r'<r, b' \leq b, \ell<j}} \sum P_{\ell}(e',r',b') P_{j-1-\ell}(e'',r-1-r',b-b')\\
+ \underset{\substack{e' \in N_G(e)\\v \in V(e')}} \sum y_u P_{j-1}(e',r-1,b)+\underset{\substack{e' \in N_G(e)\\u \in V(e')}} \sum y_v P_{j-1}(e',r-1,b). 
\end{multline*}

\noindent If $e \in E_B$, then we have
\begin{multline*}
P_j(e,r,b) = \underset{\substack{e',e'' \in N_G(e)\\ u \in V(e'), v \in V(e'')\\ r'<r, b' \leq b, \ell<j}} \sum P_{\ell}(e',r',b') P_{j-1-\ell}(e'',r-r',b-1-b') \\
+ \underset{\substack{e' \in N_G(e)\\v \in V(e')}} \sum y_u P_{j-1}(e',r,b-1)+\underset{\substack{e' \in N_G(e)\\u \in V(e')}} \sum y_v P_{j-1}(e',r,b-1).
\end{multline*}

We now show that every multilinear monomial of $P_j(e,r,b)$ corresponds to a relaxed-tree $X=(T,h)$ of size $j$ such that the $T$ has $r$ red edges and $b$ blue edges where $h$ is injective.
We prove this claim by induction on $j$. The base case is easy to verify. 
Consider the induction step. 
Suppose $e=\{u,v\} \in E_R$ (the other case is symmetric). 
Let $M$ be a multilinear monomial of $P_j(e,r,b)$ where $j>1$.
\begin{itemize}
\item Case 1: $M=y_u M'$ where $M'$ is a multilinear monomial of $P_{j-1}(e',r-1,b)$ such that $e' \in N_G(e)$.
Let $e'=\{v,w\}$.
By induction, $M'$ corresponds to a relaxed-tree $X'=(T',h')$ of size $j-1$ such that $e' \in E(X')$. 
Further, $T'$ has $r-1$ red edges and $b$ blue edges.
Also, $v \in V(h'(T'))$ (by definition) and $u \notin V(h'(T'))$ (due to the multilinearity of $M$).
Let $z'$ denote $h'(v)$.
Let $T$ denote the tree obtained from $T'$ by adding new vertex $z$ adjacent to $z'$ and the edge $\{z,z'\}$ colored red.
Let $h$ be the homomorphism obtained from $h'$ by extending the domain to include $z$ with $h(z)=u$.
Then, $X=(T,h)$ is the relaxed-tree with the desired properties that $M$ corresponds to.
\item Case 2: $M=M_1 M_2$ where $M_1$ is a multilinear monomial of $P_{j_1}(e',r',b')$ and $M_2$ is a multilinear monomial of $P_{j_2}(e'',r'',b'')$ such that $e',e'' \in N_G(e)$, $j_1,j_2 <j$, $r'' \leq r-1$ and $b'' \leq b$.
Let $e'=\{u,t\}$ and $e''=\{v,w\}$.
By induction, $M_1$ corresponds to a relaxed-tree $X_1=(T_1,h_1)$ of size $j_1$ such that $e' \in E(X_1)$. 
Further, $T_1$ has $r'$ red edges and $b'$ blue edges.
Also, $u \in V(h_1(T_1))$.
Similarly, $M_2$ corresponds to a relaxed-tree $X_2=(T_2,h_2)$ of size $j_2$ such that $e \in E(X_2)$. 
Further, $T_2$ has $r''$ red edges and $b''$ blue edges.
Also, $v \in V(h_2(X_2))$.
Without loss of generality assume  $V(T_1) \cap V(T_2)=\emptyset$ (one can achieve this property by renaming vertices). 
Let $T$ denote the graph with vertex set $V(T_1) \cup V(T_2)$ and edge set $E(T_1) \cup E(T_2) \cup \{h_1^{-1}(u),h_2^{-1}(v)\}$.
Let $h: V(T) \rightarrow V(G)$ denote the homomorphism obtained from $h_1$ and $h_2$ by extending the domain to the union of the domains of $h_1$ and $h_2$.
Due to the multilinearity of $M$, it follows that $h$ is injective.
Then, $X=(T,h)$ is the required relaxed-tree that $M$ corresponds to.
\end{itemize}

We next show that if there is a relaxed-tree $X=(T,h)$ of size $j$ with $r$ red edges, $b$ blue edges and such that $e=\{u,v\} \in E(X)$, then there is a monomial of $P_j(e,r,b)$ that corresponds to it. We show this again using induction on $j$. Suppose $e\in E_R$ (the other case is symmetric). The base case is trivial. 
Consider the induction step ($j \geq  2$). Let $a=h^{-1}(u)$, $b=h^{-1}(v)$ and $z$ denote the edge $\{a,b\}$ of $T$.
Then $T-z$ has two components $T_a$ (containing $a$) and $T_b$ (containing $b$).
Let $h_a$ and $h_b$ denote the homomorphism $h$ restricted to domains $V(T_a)$ and $V(T_b)$, respectively.
Without loss of generality let $T_a$ have at least one edge.
Suppose $T_b$ is edgeless. 
By induction, let $M'$ be the monomial corresponding to $(T_a,h_a)$.
Then, $y_v M'$ is the monomial in $P_j(e,r,b)$ corresponding to $(T,h)$.
Next, consider the case when $T_b$ is not edgeless.
Let $j_1$ and $j_2$ denote the number of edges in $T_a$ and $T_b$, respectively.
Let $r'$ and $r''$ be the number of red edges in $T_a$ and $T_b$, respectively.
Similarly, let $b'$ and $b''$ be the number of blue edges in $T_a$ and $T_b$, respectively.
Then $j-1=j_1 +j_2, r-1=r'+r'', b=b'+b''$.
Now, $X_1=(T_a,h_a)$ and $X_2=(T_b,h_b)$ are relaxed-trees with the following properties.
\begin{itemize}
\item The size of $X_1$ is $j_1$ and $T_a$ has $r'$ red edges, $b'$ blue edges and there is an edge $e'\in E(X_1)$ incident on $h_a(a)$.
\item The size of $X_2$ is $j_2$ and $T_b$ has $r''$ red edges, $b''$ blue edges and there is an edge $e''\in E(X_2)$ incident on $h_b(b)$.
\end{itemize}
By induction, there is a monomial $M_1$ in $P_{j_1}(e',r',b')$ that corresponds to $X_1$ and there is a monomial $M_2$ in $P_{j_2}(e'',r'',b'')$ that corresponds to $X_2$.
Then, the monomial $M_1M_2$ which is in $P_j(e,r,b)$ corresponds to $X$.

Finally, let $P= \underset{e \in E(G)} \sum P_k(e,\frac{k}{2},\frac{k}{2})$.
Every monomial in $P$ has degree $k+1$. Then from the arguments above, $P$ is the desired polynomial.
To compute $P$ we need to compute $P_j(e,r,b)$ for each $j,r,b \in [k]$ and each edge $e \in E(G)$.
As these polynomials can be represented as a 
polynomial-sized arithmetic circuit, the reduction runs in polynomial time.  
\end{proof}

Finally, we consider \textsc{Exact Edge Balanced Path}.

\begin{restatable}{mytheorem}{bpmld}
\label{thm:bp-fpt-rand}
\textsc{Exact Edge Balanced Path} admits a randomized $\mathcal{O}^*(2^{k})$-time algorithm.
\end{restatable}

\begin{proof}
Consider an instance $(G,k)$. 
Let $E_R$ denote the set of red edges and $E_B$ denote the set of blue edges in $G$.
We will define a polynomial $P$ over the variable set $\{y_v : v \in V(G)\}$ satisfying the following properties.
\begin{itemize}
    \item For each balanced relaxed-path $X$ of size $k$ there exists a monomial in $P$ that corresponds to $X$.  
    We say that the monomial $M$ corresponds to $X$, if $M=\underset{v \in V(X)}\prod y_v$.
    \item For each multilinear monomial $M$ in $P$, there is a balanced relaxed-path $X=(Q,h)$ of size $k$ where $h$ is injective that $M$ corresponds to.
\end{itemize}

If $P$ is such a polynomial,
then from Observation \ref{obs:relaxed-structures}, $G$ has a balanced path of size $k$ if and only if $P$ has a multilinear monomial of degree $k+1$. This way, after the construction of $P$, we reduce the problem to the {\sc Multilinear Monomial Detection} and use Proposition \ref{prop:mld}. 
In order to construct $P$, we first construct polynomials $P_j(v,r,b)$ for each $v \in V(G)$, $j \in [k]$ and $0 \leq r,b \leq \frac{k}{2}$.
Monomials of $P_j(v,r,b)$ will correspond to relaxed-paths $X=(Q,h)$ of size $j$ such that $Q$ has $r$ red edges and $b$ blue edges and $v=h(u_1)$, where $u_1$ is the first vertex of $Q$.
The construction of $P_j(v,r,b)$ is as follows. 
For a vertex $v \in V(G)$, $P_0(v,0,0) = y_v$. Also, $P_j(v,r,b)=0$ if $j \neq r+b$. For non-negative integers $r,b$ with $r+b \geq 1$ and a positive integer $j$, we have 

\begin{align*}
P_j(v,r,b) = \underset{\substack{u \in N(v) \\ \{u,v\} \in E_R}} \sum y_v P_{j-1}(u,r-1,b) + \underset{\substack{u \in N(v) \\ \{u,v\} \in E_B}} \sum y_v P_{j-1}(u,r,b-1).
\end{align*} 

We now show that every multilinear monomial of $P_j(v,r,b)$ corresponds to a relaxed-path $(Q,h)$ of size $j$ such that $Q$ has $r$ red edges, $b$ blue edges and $v$ is the vertex of $G$ onto which the first vertex of $Q$ is mapped under $h$ and $h$ is injective.
We show this using induction on $j$. The base case is easy to verify. 
Consider the induction step. 
Let $M$ be a multilinear monomial of $P_j(v,r,b)$.
Then, $M=y_v M'$ where $M'$ is a multilinear monomial of $P_{j-1}(u,r-1,b)$.
By induction hypothesis, $M'$ corresponds to a relaxed-path $X'=(Q',h')$ such that $h'$ is injective with $h'(z')=u$ where $z'$ is the starting vertex of $Q'$. Further, $Q'$ has $r-1$ red edges, $b$ blue edges or $r$ red edges, $b-1$ blue edges. 
Let $Q$ denote the path obtained from $Q'$ by adding a new vertex $z$ adjacent to $z'$ with the edge $\{z,z'\}$ colored red (in the former case) or blue (in the latter case).
Let $h$ be the homomorphism obtained from $h'$ by extending the domain to include $z$ with $h(z)=v$.
Clearly, $h$ is injective as $h'$ is injective and $v \notin V(Q')$ (due to the multilinearity of $M$).
Now, $y_v M'$ corresponds to the relaxed-path $X=(Q,h)$.
And, $X$ is of size $j$ such that $Q$ has $r$ red edges, $b$ blue edges and $v$ is the vertex of $G$ onto which the first vertex of $Q$ is mapped under $h$.

We next show that if there is a relaxed-path $X=(Q,h)$ of size $j$ with $r$ red edges, $b$ blue edges such that $v$ is the first vertex of this relaxed path, then there is a monomial of $P_j(v,r,b)$ that corresponds to it. We show this again using induction on $j$. 
Let $\{a,b\}$ be the first edge of $Q$. 
Let $v=h(a)$ and $u=h(b)$. 
Suppose $\{v,u\} \in E_R$ (the other case is symmetric). 
Consider $X'=(Q',h')$ where $V(Q')=V(Q) \setminus \{v\}$, $E(Q')=E(Q) \setminus \{\{a,b\}\}$ and $h'$ is $h$ with domain restricted to $V(Q')$.
Then, $X'$ is a relaxed-path with $r-1$ red and $b$ blue edges. 
By induction there is a monomial $M$ corresponding to it in the $P_{j-1}(u,r-1,b)$. Thus, $y_v M$ corresponds to $X$.

Finally, let $P= \underset{v \in V(G)} \sum P_k(v,\frac{k}{2},\frac{k}{2})$.
Every monomial in $P$ has degree $k+1$. Then from the arguments above, $P$ is the desired polynomial. As the polynomial $P$ can be represented as a polynomial-sized arithmetic circuit, the reduction runs in polynomial time.  
\end{proof}

\subsection{Deterministic Algorithms}
We first describe deterministic algorithms for \textsc{Exact Edge Balanced Connected Subgraph} and \textsc{Exact Edge Balanced Tree} using the color-coding technique \cite{CyganFKLMPPS15,AlonYZ95,NaorSS95}. 

Consider an instance $(G,k)$ of \textsc{Exact Edge Balanced Connected Subgraph/Tree}. Let $m=\vert E(G) \vert$ and $n=|V(G)|$.
Let $E_R$ denote the set of red edges and $E_B$ denote the set of blue edges in $G$.
Let $\sigma : E(G) \rightarrow [k]$ be a coloring of edges of $G$ and $\tau : V(G) \rightarrow [k+1]$ be a coloring of vertices of $G$.
We now define $L$-edge-colorful subgraphs and $L$-vertex-colorful subgraphs.
For $L\subseteq [k+1]$, a subgraph $H \subseteq G$ is said to be $L$-edge-colorful if $\vert E(H)\vert = \vert L \vert$, $\bigcup_{e \in E(H)} \sigma(e) =L$ and for every $e \neq e' \in E(H)$, $\sigma(e) \neq \sigma(e')$. 
Similarly, $H$ is said to be $L$-vertex-colorful if $\vert V(H)\vert = \vert L \vert$, $\bigcup_{v \in V(H)} \tau(v) =L$ and for every $u \neq v \in V(H)$, $\tau(u) \neq \tau(v)$. We describe dynamic programming algorithms to find a $[k]$-edge-colorful balanced connected subgraph and a $[k+1]$-vertex-colorful balanced tree in $G$ (if they exist) in $\mathcal{O}^*(4^{k})$ time.

\begin{mylemma}
\label{lem:colorful-bcs}
If a red-blue graph $G$ associated with coloring $\sigma : E(G) \rightarrow [k]$ has a $[k]$-colorful balanced connected subgraph then such a subgraph can be obtained in $\mathcal{O}^*(4^{k})$ time. 
\end{mylemma}

\begin{proof}
For $L\subseteq [k], e\in E(G)$ and $r,b\leq \frac{k}{2}$, define $\Lambda(e,L,r,b)$ to be 1 if there is a connected $L$-edge-colorful subgraph $H$ of $G$ containing $e$ such that $\vert E(H) \cap E_R \vert =r$, $\vert E(H) \cap E_B \vert =b$ and 0 otherwise. 
Clearly, $G$ has a $[k]$-edge-colorful balanced connected subgraph if and only if there is an edge $e\in E(G)$ such that $\Lambda(e,[k], \frac{k}{2}, \frac{k}{2})=1$. 
First observe that for $L \subseteq [k]$, $\vert L \vert=1$ ($L=\{i\}$), $\Lambda(e,i,1,0)=1$ if $e$ is in $E_R$ and 0 otherwise. 
Similarly, $\Lambda(e,i,0,1)=1$ if $e$ is in $E_B$ and 0 otherwise. 
Therefore, the entries $\Lambda(e,i,r,b)$ for $i\in [k],  e\in E(G),  r+b \leq 1$ can be filled in polynomial time. 
Also, $\Lambda(e,L,r,b)=0$ if $|L| \neq r+b$ or $r=0$ with $e \in E_R$ or $b=0$ with $e \in E_B$.

Next, consider $\Lambda(e, L, r, b)$ where $|L|>1$ (and $r+b >1$). We claim that $\Lambda(e, L, r, b)=1$ if and only if one of the following is true. Let $e=\{u,v\}$.
\begin{itemize}
    \item $\Lambda(e', L \backslash \sigma(\{u,v\}), r-1,b)=1$ where $e'\in N(e)$ and $e \in E_R$.
     \item $\Lambda(e', L \backslash \sigma(\{u,v\}), r,b-1)=1$ where $e'\in N(e)$ and $e \in E_B$.
    \item $\Lambda(e',L', r', b')=1$ and $\Lambda(e'',L\backslash (L' \cup \sigma(\{u,v\})), r-r'-1, b-b')=1$ where $e',e''\in N(e)$, $L'\subseteq L \backslash \sigma(\{u,v\})$, $r' \leq r-1$, $b'\leq b$ and $e \in E_R$. 
  \item $\Lambda(e',L', r', b')=1$ and $\Lambda(e'',L\backslash (L' \cup \sigma(\{u,v\})), r-r', b-b'-1)=1$ where $e',e''\in N(e)$, $L'\subseteq L \backslash \sigma(\{u,v\})$, $r' \leq r$, $b'\leq b-1$ and $e \in E_B$. 
\end{itemize}

Observe that if $H$ is a solution to one of the above, then adding $e$ to $H$ is required subgraph that makes $\Lambda(e, L, r, b)$ as 1.
Conversely, suppose $\Lambda(e, L, r, b)=1$. Then $G$ has a $L$-colorful connected subgraph $H$ such that $e\in E(H)$, $\vert E(H) \cap E_R \vert = r$ and $\vert E(H) \cap E_B \vert = b$. 
Without loss of generality, let $e=\{u,v\} \in E_R$. 
Consider the graph $H'$ obtained from $H$ by deleting $e$. 
Then, $H'$ has at most two components.
If $H'$ is connected, then either $\Lambda(e', L \backslash \sigma(\{u,v\}), r-1,b)=1$ or $\Lambda(e', L \backslash \sigma(\{u,v\}), r,b-1)=1$ for some edge $e' \in E(H') \cap N(e)$. 
Otherwise, $H$ has two components and let $H_u$ be the one containing $u$ and $H_v$ be the one containing $v$. 
If $H_u$ has no edge, then $\Lambda(e', L \backslash \sigma(e), r-1,b)=1$ for some $e'\in N(e) \cap E(H_v)$.
Similarly, if $H_v$ has no edge, then $\Lambda(e', L \backslash \sigma(e), r-1,b)=1$ for some $e'\in N(e) \cap E(H_u)$.
Otherwise, let $L'$ denote the set of colors $\sigma(e)$ of edges $e\in H_{u}$. Then $L \backslash (L' \cup \sigma(\{u,v\}))$ is the set of colors of edges in $H_{v}$. 
Also $\Lambda(e',L',r',b')=1$ where $r' = \vert E_R \cap E(H_{u})\vert$ and $b' = \vert E_B \cap E(H_{u}) \vert$ and $\Lambda(e'', L \backslash (L' \cup \sigma(e)), r-r'-1, b-b')=1$. 

The time taken to compute an entry $\Lambda(e,L,r,b)$ is $m k^2 2^{k}$ as we need to go over choices $e',r',b'$ and $L' \subseteq L$ where $\vert L \vert \leq k$, $\vert E(G) \vert =m$ and $r',b' \leq \frac{k}{2}$. 
As the number of entries $\Lambda(e,L,r,b)$ is at most $m k^2 2^{k}$, the  running time of the algorithm is $\mathcal{O}^*(4^{k})$. 
\end{proof}

Standard derandomization techniques using \textit{perfect hash families} \cite{CyganFKLMPPS15,AlonYZ95,NaorSS95} leads to the following result.

\begin{restatable}{mytheorem}{bcscolor}
\label{thm:bcs-fpt}
\textsc{Exact Edge Balanced Connected Subgraph} can be solved in $\mathcal{O}^*((4e)^{k})$ time.
\end{restatable}
\begin{proof}
First, we use the results of \cite{CyganFKLMPPS15,AlonYZ95,NaorSS95} to construct a family $\mathcal{F}_{m,k}$ of coloring functions $\sigma:E(G) \rightarrow [k]$ of size $e^{k} {k}^{\mathcal{O}(\log k)} \log m$ in $e^{k} {k}^{\mathcal{O}(\log k)} m \log m$ time satisfying the following property: for every set $E' \subseteq E(G)$ of size $k$, there is a function $\sigma \in \mathcal{F}_{m,k}$ such that $\sigma(e) \neq \sigma(e')$ for any two distinct edges $e,e' \in E'$. Now, for each function in $\mathcal{F}_{m,k}$, we use the dynamic programming algorithm given in Lemma \ref{lem:colorful-bcs} to find a $[k]$-edge-colorful balanced connected subgraph, if one exists.
The properties of $\mathcal{F}_{m,k}$ ensure that, if there exists a balanced connected subgraph $H\subseteq G$ on $k$ edges, then there exists $f\in \mathcal{F}$ that is injective on $E(H)$ and, consequently, the
algorithm finds a $[k]$-edge-colorful balanced connected subgraph. Hence, we obtain a deterministic algorithm which can solve \textsc{Exact Edge Balanced Connected Subgraph} in $\mathcal{O}^*((4e)^{k})$ time.  
\end{proof}

Next, we show that \textsc{Exact Edge Balanced Tree} can be solved in $\mathcal{O}^*((4e)^{k})$ time.

\begin{mylemma}
\label{lem:colorful-bt}
Given a red-blue graph $G$ associated with coloring $\tau: V(G) \rightarrow [k+1]$, one can determine if $G$ has a $[k+1]$-vertex-colorful balanced tree in $\mathcal{O}^*(4^k)$ time.
\end{mylemma}
\begin{proof}
For $L\subseteq [k+1], e\in E(G)$ and $r,b\leq \frac{k}{2}$, define $\Lambda(e,L,r,b)$ to be 1 if there is an $L$-vertex-colorful tree $T$ of $G$ containing $e$ such that $\vert E(T) \cap E_R \vert =r$, $\vert E(T) \cap E_B \vert =b$ and 0 otherwise. 
Clearly, $G$ has a $[k+1]$-colorful balanced tree if and only if $\exists e\in E(G)$ such that $\Lambda(e,[k+1], \frac{k}{2}, \frac{k}{2})=1$.
Observe that $\Lambda(e,L,r,b)=0$ if $|L| \neq r+b+1$ and $\Lambda(e,L,1,0)$ is 1 if and only if $L=\{i,j\}$ where $e=\{u,v\} \in E_R$ with $\tau(u)=i$ and $\tau(v)=j$. Similarly, $\Lambda(e,L,0,1)$ is 1 if and only if $L=\{i,j\}$ where $e=\{u,v\} \in E_B$ with $\tau(u)=i$ and $\tau(v)=j$.

Consider $\Lambda(e, L, r, b)$ where $r+b \geq 1$. Then, $\Lambda(e, L, r, b)=1$ if and only if one of the following holds. Let $e=\{u,v\}$.
\begin{itemize}
\item $\Lambda(e', L \backslash \tau(u), r-1,b)=1$ where $e'\in N(e)$ and $e \in E_R$.
\item $\Lambda(e', L \backslash \tau(v), r-1,b)=1$ where $e'\in N(e)$ and $e \in E_R$.
\item $\Lambda(e', L \backslash \tau(u), r,b-1)=1$ where $e'\in N(e)$ and $e \in E_B$.
\item $\Lambda(e', L \backslash \tau(v), r,b-1)=1$ where $e'\in N(e)$ and $e \in E_B$.
    \item $\Lambda(e',L', r', b')=1$ and $\Lambda(e'',L'', r-r'-1, b-b')=1$ where $e',e''\in N(e)$, $L'\subseteq L \backslash \tau(u)$, $L''=L \setminus L'$, $r' \leq r-1$, $b'\leq b$ and $e \in E_R$. 
   \item $\Lambda(e',L', r', b')=1$ and $\Lambda(e'',L'', r-r', b-b'-1)=1$ where $e',e''\in N(e)$, $L'\subseteq L \backslash \tau(u)$, $L''=L \setminus L'$, $r' \leq r$, $b'\leq b-1$ and $e \in E_B$. 
\end{itemize}
The time taken to compute an entry $\Lambda(e,L,r,b)$ is $m k^2 2^{k+1}$ as we need to go over choices $e',r',b'$ and $L' \subseteq L$ where $\vert L \vert \leq k+1$, $\vert E(G) \vert =m$ and $r',b' \leq \frac{k}{2}$. 
As the number of entries $\Lambda(e,L,r,b)$ is at most $m k^2 2^{k+1}$, the  running time of the algorithm for a given $T$ is $\mathcal{O}^*(4^{k})$. 
\end{proof}

Similar to Theorem \ref{thm:bcs-fpt}, Lemma \ref{lem:colorful-bt} along with derandomization using perfect hash families lead to the following result. 

\begin{restatable}{mytheorem}{btcolor}
\label{thm:bt-fpt}
\textsc{Exact Edge Balanced Tree} can be solved in $\mathcal{O}^*((4e)^{k})$ time.
\end{restatable}

As the reader would have already observed, a simpler dynamic programming algorithm along with derandomization using perfect hash families results in an algorithm for \textsc{Exact Edge Balanced Path} that runs in $\mathcal{O}^*((2e)^{k})$ time.
Subsequently, we describe a faster algorithm using representative sets.  
We begin with some definitions and results related to representative sets.
For a finite set $U$, let ${U \choose p}$ denote the set of all subsets of size $p$ of $U$. 
Given two families $\mathcal{S}_1, \mathcal{S}_2 \subseteq 2^{U}$, the {\em convolution} of $\mathcal{S}_1$ and $\mathcal{S}_2$ is the new family defined as $\mathcal{S}_1 * \mathcal{S}_2= \{X \cup Y \mid X \in \mathcal{S}_1, Y \in \mathcal{S}_2, X \cap Y=\emptyset \}$.

\begin{restatable}{mydefinition}{repset}
Let $U$ be a set and $\mathcal{S} \subseteq {U \choose p}$.
A subfamily $\widehat{\mathcal{S}} \subseteq \mathcal{S}$ is said to $q$-represent $\mathcal{S}$ (denoted as $\widehat{\mathcal{S}} \subseteq^q_{rep} \mathcal{S}$) if for every set $Y \subseteq U$ of size at most $q$ such that there is a set $X \in \mathcal{S}$ with $X \cap Y=\emptyset$, there is a set $\widehat{X} \in \widehat{\mathcal{S}}$ with $\widehat{X} \cap Y=\emptyset$.
If $\widehat{\mathcal{S}} \subseteq^q_{rep} \mathcal{S}$, then $\widehat{\mathcal{S}}$ is called a {\em $q$-representative family} for $\mathcal{S}$.
\end{restatable}

Representative families (also called representative sets) are transitive and have nice union and convolution properties (Proposition \ref{prop:rep-sets-prop}). 

\begin{restatable}{myproposition}{repsetprop}
\label{prop:rep-sets-prop}
{\em \cite[Lemmas 12.26, 12.27 and 12.28]{CyganFKLMPPS15}}
Let $U$ be a fnite set.
\begin{enumerate}
\item Let $\mathcal{S}_1,  \mathcal{S}_2 \subseteq {U \choose p}$. If $\widehat{\mathcal{S}}_1 \subseteq^q_{rep} \mathcal{S}_1$ and $\widehat{\mathcal{S}}_2 \subseteq^q_{rep} \mathcal{S}_2$, then $\widehat{\mathcal{S}}_1 \cup \widehat{\mathcal{S}}_2 \subseteq^q_{rep} \mathcal{S}_1 \cup \mathcal{S}_2$.
\item Let $\mathcal{S} \subseteq {U \choose p}$. If $\widehat{\mathcal{S}} \subseteq^q_{rep} \mathcal{S}'$ and $\mathcal{S}' \subseteq^q_{rep} \mathcal{S}$, then $\widehat{\mathcal{S}} \subseteq^q_{rep} \mathcal{S}$.
\item Let $\mathcal{S}_1 \subseteq {U \choose p_1}$, $\mathcal{S}_2 \subseteq {U \choose p_2}$. If $\widehat{\mathcal{S}}_1 \subseteq^{k-p_1}_{rep} \mathcal{S}_1$ and $\widehat{\mathcal{S}}_2 \subseteq^{k-p_2}_{rep} \mathcal{S}_2$, then $\widehat{\mathcal{S}}_1 * \widehat{\mathcal{S}}_2 \subseteq^{k-p_1-p_2}_{rep} \mathcal{S}_1 * \mathcal{S}_2$.
\end{enumerate}
\end{restatable}

A classical result due to Bollob\'{a}s states that small representative families exist \cite{Bollobas65} and Proposition \ref{prop:rep-sets-comp} \cite{CyganFKLMPPS15,FominLPS16,ShachnaiZ16} shows that such families can be efficiently computed. 

\begin{restatable}{myproposition}{repsetcomp}
{\em \cite{CyganFKLMPPS15,FominLPS16,ShachnaiZ16}}
\label{prop:rep-sets-comp}
There is an algorithm that given a family $\mathcal{S} \subseteq {U \choose p}$, a rational $0 < x < 1$ and integers $p$, $k \geq p$, computes  $\widehat{\mathcal{S}} \subseteq^{k-p}_{rep} \mathcal{S}$ of size at most
$x^{-p} (1-x)^{-(k-p)}2^{o(k)}$ in $|\mathcal{S}| (1-x)^{-(k-p)} 2^{o(k)}$ time.
Further, if $|\mathcal{S}|=\mathcal{O}^*(x^{-p} (1-x)^{-(k-p)}2^{o(k)})$ and  $x=\frac{p}{2k-p}$, then the construction of $\widehat{\mathcal{S}}$ takes 
$\mathcal{O}^*(\frac{(2k-p)^{2k-p}}{p^p (2k-2p)^{2k-2p}}2^{o(k)})$ time, 
moreover, this running time is maximized when $p=(1-\frac{1}{\sqrt{5}})k$ and is $\mathcal{O}^*(\phi^{2k + o(k)})$ (which is $\mathcal{O}^*(2.619^k)$) where $\phi$ is the golden ratio $\frac{1+\sqrt{5}}{2}$.
\end{restatable}

Now, we are ready to describe an algorithm for \textsc{Exact Edge Balanced Path} using representative sets. 

\begin{restatable}{mytheorem}{bprepset}
\label{thm:bp-fpt-rep}
\textsc{Exact Edge Balanced Path} can be solved in $\mathcal{O}^*(2.619^{k})$ time.
\end{restatable}
\begin{proof}
Consider an instance $(G,k)$.
Let $E_R$ and $E_B$ denote the sets of red and blue edges of $G$. 
For a pair of vertices $u,v \in V(G)$ and non-negative integers $r$ and $b$ with $r+b\geq 1$, define the family $\mathcal{P}^{(r,b)}_{uv}$ as follows.
\begin{multline*}
\mathcal{P}^{(r,b)}_{uv}= \{X \mid X \subseteq V(G), ~|X|=r+b+1 \text{ and there is a path P from } u \text{to } v  \text{ with } \\ 
 V(P)=X, ~|E_R \cap E(P)| =r \text{ and } |E_B \cap E(P)| =b  \}.
\end{multline*}

\noindent Now, it suffices to determine if $\mathcal{P}^{(\frac{k}{2},\frac{k}{2})}_{uv}$ is non-empty for some $u,v \in V(G)$. 
The families $\mathcal{P}^{(r,b)}_{uv}$ can be computed using the following formula.
For $r+b=1$, 
\begin{equation*}
    \mathcal{P}^{(1,0)}_{uv}=
    \begin{cases}
      \{\{u,v\}\}, & \text{if}\ \{u,v\} \in E_R \\
      \emptyset, & \text{otherwise}
    \end{cases}
    \text{ and }\
     \mathcal{P}^{(0,1)}_{uv}=
    \begin{cases}
      \{\{u,v\}\}, & \text{if}\ \{u,v\} \in E_B \\
      \emptyset, & \text{otherwise}
    \end{cases}
  \end{equation*}

\noindent and, for $r+b>1$,
\begin{equation*}
    \mathcal{P}^{(r,b)}_{uv}=(\underset{\{w,v\} \in E_R}\bigcup (\mathcal{P}^{(r-1,b)}_{uw} * \{\{v\}\})) \bigcup (\underset{\{w,v\} \in E_B}\bigcup (\mathcal{P}^{(r,b-1)}_{uw}* \{\{v\}\}))
  \end{equation*}
  
The base case is easy to verify. 
Consider a path $P$ such that $V(P) \in  \mathcal{P}^{(r,b)}_{uv}$. 
Then, $P$ has a subpath $P'$ from $u$ to a neighbour $w$ of $v$. 
If $\{w,v\}\in E_R$, then $P'$ has $r-1$ red edges and $b$ blue edges. 
Further, $V(P') \in \mathcal{P}^{(r-1,b)}_{uw}$ and $V(P)=V(P') \cup \{v\}$ where $v \notin V(P')$. 
Otherwise,  $\{w,v\}\in E_B$ and $P'$ has $r$ red edges and $b-1$ blue edges. 
Now, $V(P') \in \mathcal{P}^{(r,b-1)}_{uw}$ and $V(P)=V(P') \cup \{v\}$ where $v \notin V(P')$. 
On the other hand, for any element $X \in \mathcal{P}^{(r,b-1)}_{uw}$ such that there is a vertex $v$ with $v \notin X$ and $\{w,v\} \in E_B$, we have $X \cup \{v\} \in \mathcal{P}^{(r,b)}_{uv}$.  
Similarly, for any element $X \in \mathcal{P}^{(r-1,b)}_{uw}$ such that there is a vertex $v$ with $v \notin X$ and $\{w,v\} \in E_R$, we have $X \cup \{v\} \in \mathcal{P}^{(r,b)}_{uv}$.  
This justifies  the formula given for the computation of  $\mathcal{P}^{(r,b)}_{uv}$.

Clearly, a naive computation of $\mathcal{P}^{(r,b)}_{uv}$ is not guaranteed to result in an \FPT\ (in $k$) algorithm.
Therefore, instead of computing $\mathcal{P}^{(r,b)}_{uv}$, we only compute $\widehat{\mathcal{P}}^{(r,b)}_{uv} \subseteq^{k-(r+b)}_{rep} \mathcal{P}^{(r,b)}_{uv}$ and use the fact that 
$\widehat{\mathcal{P}}^{(\frac{k}{2},\frac{k}{2})}_{uv} \subseteq^0_{rep} \mathcal{P}^{(\frac{k}{2},\frac{k}{2})}_{uv}$.
If $\mathcal{P}^{(\frac{k}{2},\frac{k}{2})}_{uv}$ is nonempty, then it contains a set $X$ that is disjoint with $\emptyset$. As $\widehat{\mathcal{P}}^{(\frac{k}{2},\frac{k}{2})}_{uv} \subseteq^0_{rep} \mathcal{P}^{(\frac{k}{2},\frac{k}{2})}_{uv}$, it follows that $\widehat{\mathcal{P}}^{(\frac{k}{2},\frac{k}{2})}_{uv}$ also has a set $\widehat{X}$ that is disjoint with $\emptyset$. In other words, if $\mathcal{P}^{(\frac{k}{2},\frac{k}{2})}_{uv}$ is nonempty, then $\widehat{\mathcal{P}}^{(\frac{k}{2},\frac{k}{2})}_{uv}$ is also non-empty.

Now, we describe a dynamic programming algorithm to compute $\widehat{\mathcal{P}}^{(\frac{k}{2},\frac{k}{2})}_{uv}$ for every $u,v \in V(G)$.
For $r+b=1$, set $\widehat{\mathcal{P}}^{(1,0)}_{uv}=\mathcal{P}^{(1,0)}_{uv}$ and $\widehat{\mathcal{P}}^{(0,1)}_{uv}=\mathcal{P}^{(0,1)}_{uv}$.
Clearly, $\widehat{\mathcal{P}}^{(1,0)}_{uv} \subseteq^{k-1}_{rep} \mathcal{P}^{(1,0)}_{uv}$ and $\widehat{\mathcal{P}}^{(0,1)}_{uv} \subseteq^{k-1}_{rep}\mathcal{P}^{(0,1)}_{uv}$.
Further, $|\widehat{\mathcal{P}}^{(1,0)}_{uv}|,|\widehat{\mathcal{P}}^{(0,1)}_{uv}| \leq 1$ and this computation is polynomial time.
Now, we proceed to computing $\widehat{\mathcal{P}}^{(r,b)}_{uv} \subseteq^{k-(r+b)}_{rep} \mathcal{P}^{(r,b)}_{uv}$ for $k \geq r+b >1$ in the increasing order of $r+b$.  
Towards this, we compute a new  family $\widetilde{\mathcal{P}}^{(r,b)}_{uv}$ as follows.

\begin{equation*}
\widetilde{\mathcal{P}}^{(r,b)}_{uv}=(\underset{\{w,v\} \in E_R}\bigcup (\widehat{\mathcal{P}}^{(r-1,b)}_{uw}* \{\{v\}\})) \bigcup (\underset{\{w,v\} \in E_B}\bigcup \widehat{\mathcal{P}}^{(r,b-1)}_{uw} * \{\{v\}\}))
\end{equation*}
Using the union and convolution properties of Propositions \ref{prop:rep-sets-prop}, we have the following properties.\\
\begin{itemize}
\item $(\widehat{\mathcal{P}}^{(r-1,b)}_{uw}* \{\{v\}\}) \subseteq^{k-(r+b)}_{rep} (\mathcal{P}^{(r-1,b)}_{uw} * \{\{v\}\})$\\
\item $(\widehat{\mathcal{P}}^{(r,b-1)}_{uw}* \{\{v\}\}) \subseteq^{k-(r+b)}_{rep} (\mathcal{P}^{(r,b-1)}_{uw} * \{\{v\}\})$\\
\item $\underset{\{w,v\} \in E_R}\bigcup (\widehat{\mathcal{P}}^{(r-1,b)}_{uw}* \{\{v\}\}) \subseteq^{k-(r+b)}_{rep} \underset{\{w,v\} \in E_R}\bigcup (\mathcal{P}^{(r-1,b)}_{uw} * \{\{v\}\})$\\
\item $\underset{\{w,v\} \in E_B}\bigcup (\widehat{\mathcal{P}}^{(r,b-1)}_{uw}* \{\{v\}\}) \subseteq^{k-(r+b)}_{rep} \underset{\{w,v\} \in E_B}\bigcup (\mathcal{P}^{(r,b-1)}_{uw} * \{\{v\}\})$\\
\end{itemize}
Now, once again by the union property of Proposition \ref{prop:rep-sets-prop}, we have $\widetilde{\mathcal{P}}^{(r,b)}_{uv} \subseteq^{k-(r+b)}_{rep} \mathcal{P}^{(r,b)}_{uv}$.
Further, $|\widetilde{\mathcal{P}}^{(r,b)}_{uv}| = \mathcal{O}^*(|\widehat{\mathcal{P}}^{(r-1,b)}_{uv}| + |\widehat{\mathcal{P}}^{(r,b-1)}_{uv}|)$.
Then, we use Proposition \ref{prop:rep-sets-comp} to compute a family $\widehat{\mathcal{P}}^{(r,b)}_{uv} \subseteq^{k-(r+b)}_{rep} \widetilde{\mathcal{P}}^{(r,b)}_{uv}$.
By the transitivity property of Proposition \ref{prop:rep-sets-prop}, it follows that $\widehat{\mathcal{P}}^{(r,b)}_{uv} \subseteq^{k-(r+b)}_{rep} \mathcal{P}^{(r,b)}_{uv}$. The time taken to compute the families $\widehat{\mathcal{P}}^{(r,b)}_{uv}$ for $u,v \in V(G)$ and $r,b \leq \frac{k}{2}$ is $\mathcal{O}^*(2.619^k)$ from Proposition \ref{prop:rep-sets-comp} by substituting $r+b+1$ for $p$ and $k+1$ for $k$. Thus, the overall running time of the algorithm is $\mathcal{O}^*(2.619^k)$.
\end{proof}

\section{Concluding Remarks}
\label{sec:concl}
To summarize our work, we study the complexity of finding balanced connected subgraphs, trees and paths in red-blue graphs. We give fixed-parameter tractability results using color coding, representative sets and reductions to \textsc{Multilinear Monomial Detection}. En route, we give combinatorial results on the existence of small balanced connected subgraphs, trees and paths. We observe that these results also extend to vertex-balanced connected subgraphs, trees and paths. As a result the algorithms described in this work also generalize to solve the vertex-analogue of the problems. Note that using line graphs, one can reduce \textsc{Edge Balanced Connected Subgraph} to \textsc{Vertex Balanced Connected Subgraph}, however, when the solution is required to be a path or a tree, this reduction is not useful. Determining the complexity of finding other balanced substructures is an interesting direction of research. It is well-known (by a observation made by Erd\H{o}s) and easy to verify that a monochromatic spanning tree exists in any red-blue complete graph. This fact has been generalized in several directions. Gy{\'a}rfas \cite{Gyarfas77} and F{\"u}redi \cite{Furedi81} independently showed that every $r$-edge-coloring of the complete graph on $n$ vertices results in a monochromatic connected subgraph of size at least $n/(r-1)$. Bollob{\'a}s and Gy{\'a}rfas \cite{BollobasG08} studied monochromatic 2-connected subgraphs and showed that every red-blue complete graph on $n \geq 5$ vertices has a monochromatic 2-connected subgraph with at least $n-2$ vertices. They also dealt with questions on the existence of monochromatic $q$-connected subgraphs. Variants such as large monochromatic components of small diameter have also attracted attention recently \cite{CarlsonMPR22}. A similar combinatorial (and algorithmic) study on balanced spanning trees, balanced spanning connected subgraphs and balanced $q$-connected subgraphs is worth investigating. Also, studying the problems on graphs that are colored using more than two colors and on colored weighted graphs are next natural questions in this context. \\

\bibliography{refs}
\end{document}